\documentclass[letterpaper, 10pt, conference]{Classes/ieeeconf}
\IEEEoverridecommandlockouts
\usepackage{Packages/decar-common}
\usepackage{Packages/decar-dynamics}
\usepackage{Packages/decar-lie}
\usepackage{Packages/decar-post}
\usepackage{Packages/myFormat}

\usepackage{Packages/DrawPageMargin}
\usepackage{lipsum}                % for dummy text only

\title{\LARGE \bf Passivity-Based Gain-Scheduled Control with Scheduling Matrices}

\author{Sepehr Moalemi$^{1}$ and James Richard Forbes$^{2}$
    \thanks{$^{1}$Sepehr Moalemi is a M.Sc. student with the Department of Mechanical Engineering, McGill University, 817 Sherbrooke St. W., Montreal, QC H3A 0C3, Canada {\tt\small sepehr.moalemi@mail.mcgill.ca}}
    \thanks{$^{2}$James Richard Forbes is an Associate Professor and William Dawson Scholar with the Department of Mechanical Engineering, McGill University, 817 Sherbrooke St. W., Montreal, QC H3A 0C3, Canada {\tt\small james.richard.forbes@mcgill.ca}}%
}

\begin{document}
    \maketitle
    \thispagestyle{empty}
    \pagestyle{empty}

    % ensure all superscripts / subscripts at the same elevation
    \fontdimen16\textfont2=\fontdimen17\textfont2
    \fontdimen13\textfont2=5pt

    % abstract
\begin{abstract}
    This paper considers gain-scheduling of very strictly passive (VSP) subcontrollers using scheduling matrices. The use of scheduling matrices, over scalar scheduling signals, realizes greater design freedom, which in turn can improve closed-loop performance. The form and properties of the scheduling matrices such that the overall gain-scheduled controller is VSP are explicitly discussed. The proposed gain-scheduled VSP controller is used to control a rigid two-link robot subject to model uncertainty where robust input-output stability is assured via the passivity theorem. Numerical simulation results highlight the greater design freedom, resulting in improved performance, when scheduling matrices are used over scalar scheduled signals.
\end{abstract}

% keywords
    \section{Introduction}
Input-output stability theorems, such as the passivity, small gain, and conic sector theorems, have been widely used to analyze and guarantee the \(\mathcal{L}_2\)-stability of feedback interconnections. The passivity theorem ensures closed-loop \(\mathcal{L}_2\)-stability of a passive system connected in a negative feedback interconnection with a very strictly passive (VSP) system \cite{Brogliato}. The systems within the feedback loop are permitted to be time-varying or nonlinear. Moreover, precise knowledge of the system parameters is not required so long as the systems remain passive and VSP, respectively, in the face of model uncertainty. 

There are many well established linear control design methods available, such as \(\mathcal{H}_2\) and \(\mathcal{H}_{\infty}\) optimal control \cite{gang_of_four}. As such, when controlling a nonlinear system, often linear controllers are designed using a linearized model of the system about a linearization point. However, the linearized model may not capture the full dynamics of the system across a wide range of operating conditions. As a result, a controller synthesized about one linearization point may not realize adequate closed-loop performance across the same wide range of operating conditions. Gain-scheduled control is a nonlinear control technique where a set of subcontrollers are designed about multiple linearization points, and are gain-scheduled in such a way that realizes acceptable performance. More recently, the stability of gain-scheduled controllers has been studied using passivity, conicity, and dissipativity theory. In \cite{Damaren_passive_map, Forbes_Damaren}, a gain-scheduled controller composed of strictly positive real (SPR) subcontrollers is shown to be input strictly passive (ISP) \cite{Damaren_passive_map} and VSP \cite{Forbes_Damaren} provided the gain-scheduling architecture is of a specific form. An alternative passivity-based gain-scheduling architecture accounting for actuator saturation is presented in \cite{Walsh_saturation}. Gain-scheduling VSP controllers with affine dependence on plant parameters is presented in \cite{Alex_LPV}. Gain-scheduling conic systems, and relying on the conic sector theorem to ensure \(\mathcal{L}_2\)-stability of the closed-loop system, is considered in \cite{Ryan_conic}. In \cite{QSR}, QSR-dissipative properties of non-square QSR-dissipative systems are shown to be preserved under the same gain-scheduling architecture of \cite{Damaren_passive_map, Forbes_Damaren}.

The gain-scheduling techniques in \cite{Damaren_passive_map, Ryan_conic, Forbes_Damaren, Walsh_saturation, Alex_LPV, Forbes_thesis, QSR} all consider \emph{scalar scheduling signals}. In \cite{Damaren_GS_Row_Signals}, the notion of extended passive systems was shown using a row of scalar scheduling signals. Scalar scheduling signals affect the entire input-output map of the subcontrollers. When controlling a multiple-input multiple-output (MIMO) system, each control variable may require different gain scheduling to ensure acceptable performance. Additionally, it might be natural to gain-schedule based on two or more independent exogenous variables. This motivates the use of \emph{scheduling matrices} that effectively introduce more scheduling parameters to allow for additional flexibility in the scheduling of the subcontrollers. The novel contribution of this work is to extend the gain-scheduling theory in \cite{Damaren_passive_map, Forbes_Damaren} to the case of scheduling matrices. To highlight the efficacy of matrix scheduling signals, the control of a rigid two-link robot is considered. Linear VSP controllers, which take the form of SPR transfer matrices, are designed and gain-scheduled using the proposed scheduling matrices, which is compared to the scalar scheduling approach of  \cite{Damaren_passive_map, Forbes_Damaren}. The SPR subcontrollers are designed as per \cite{Damaren_passive_map, Benhabib} using the solution to the linear quadratic regulator (LQR) problem in concert with the Kalman-Yakubovich-Popov (KYP) lemma \cite{KYP}.

The remainder of this paper is as follows. Notation and preliminaries are discussed in \Cref{sec: preliminaries}. The gain-scheduling architecture is  presented in \Cref{sec: problem formulation}. Two novel theorems related to the passivity properties of the gain-scheduled system with scheduling matrices are discussed in \Cref{sec: main results}. A detailed application, complete with a discussion of controller design, is presented in \Cref{sec: simulation}, followed by closing remarks in \Cref{sec: closing remarks}.
    \section{Preliminaries} \label{sec: preliminaries}
    \subsection{Notation}
    Scalars are denoted \(\alpha \in \mathbb{R}\), matrices are denoted \(\mbf{A} \in \mathbb{R}^{n \times m}\), and column matrices are denoted \(\mbf{v} \in \mathbb{R}^{n}\). Operators are denoted by \(\bm{\mathcal{G}}\), and sets are denoted by \(\mathcal{F}\). The maximum eigenvalue and singular value of \(\mbf{A}\) are denoted as \(\lambda_{\max}(\mbf{A})\) and \(\sigma_{\max}(\mbf{A})\), respectively. A positive definite matrix is denoted as \(\mbf{A} \succ 0\).
    The notation $\diag(\cdot)$ denotes a block diagonal matrix containing its arguments. The identity and zero matrices are \(\eye\) and \(\mbf{0} \), respectively.
    \subsection{Definitions}
    \begin{definition}[Induced Matrix Norm {\cite[Section 2.7]{Zhou_Robust_Control}}]
        Given the matrix \(\mbf{A}\in \mathbb{R}^{m \times n}\), the matrix norm induced by a vector \(p\)-norm is defined as \(\mleft\| \mbf{A} \mright\|_{p} = \sup_{\mbf{x} \neq \mbf{0}} \mleft\| \mbf{A}\mbf{x} \mright\|_{p}/\mleft\| \mbf{x} \mright\|_{p}\). For \(p = 2\), it follows that \(\mleft\| \mbf{A} \mright\|_{2} = \sqrt{\lambda_{\max}(\mbf{A}^{\trans}\mbf{A})} = \sigma_{\max}(\mbf{A})\).
    \end{definition}
    \vspace{2pt}
    \begin{definition}[Truncated Signal {\cite[157-158]{Marquez}}]
        For a signal \(\mbf{u} : \mathbb{R}_{\geq 0} \to \mathbb{R}^{n}\), the truncated signal, \(\mbf{u}_T\), is defined as \(\mbf{u}_T(t) = \mbf{u}(t)\) for \(t \in \mleft[ 0, T \mright]\) and \(\mbf{u}_T(t) = \mbf{0}\) for \(t > T \in \mathbb{R}_{\geq 0}\).
    \end{definition}
    \vspace{2pt}
    \begin{definition}[Truncated Inner Product {\cite[204]{Marquez}}]
        For signals \(\mbf{u}, \mbf{y} : \mathbb{R}_{\geq 0} \to \mathbb{R}^{n}\), the truncated  inner product is defined as \(\langle \mbf{u}, \mbf{y}\rangle_T = \langle\mbf{u}_T, \mbf{y}_T\rangle = \int_{0}^{T}\hspace{-2pt} \mbf{u}^{\trans}\hspace{-1pt}(t) \mbf{y}(t)\,\dt, \forall T \in \mathbb{R}_{\geq 0}\).
    \end{definition}
    \vspace{2pt}
    \begin{definition}[$\mathcal{L}_p$ Signal Spaces {\cite[156-158]{Marquez}}] \label{def: Lp spaces}
        Given a piecewise continuous signal \(\mbf{u} : \mathbb{R}_{\geq 0} \to \mathbb{R}^{n} \), \(\mbf{u} \in \mathcal{L}_{2e}\) if \(\mleft\| \mbf{u} \mright\|_{2T} = \sqrt{\langle \mbf{u}, \mbf{u}\rangle_T} < \infty\), \(\forall T \in \mathbb{R}_{\geq0}\). Additionally, \(\mbf{u} \in \mathcal{L}_{\infty}\) if \(\mleft\| \mbf{u} \mright\|_{\infty} = \sup_{t \in \mathbb{R}_{\geq0}} \max_{i=1, \ldots, n} \lvert u_i(t) \rvert  < \infty\). 
    \end{definition}
    \vspace{2pt}
    \begin{definition}[Very Strictly Passive (VSP) {\cite[229]{Marquez}}] \label{def: VSP} 
        A square system with input \(\mbf{u} \in \mathcal{L}_{2e}\) and output \(\mbf{y} \in \mathcal{L}_{2e}\) mapped through the operator \(\bm{\mathcal{G}} : \mathcal{L}_{2e} \to \mathcal{L}_{2e}\) is VSP if there exists constants \(\beta \in \mathbb{R}_{\leq 0} \) and \(\delta, \varepsilon \in \mathbb{R}_{>0} \) such that
        \begin{equation*}
                \langle \mbf{u}, \mbf{y}\rangle_T 
                \geq
                \beta + \delta \mleft\| \mbf{u} \mright\|^{2}_{2T} + \varepsilon \mleft\| \mbf{y} \mright\|^{2}_{2T}
                , \quad \forall \mbf{u} \in \mathcal{L}_{2e}, \quad \forall T \in \mathbb{R}_{\geq0}.
        \end{equation*}
        The system is passive if \(\delta = \varepsilon = 0\), input strictly passive (ISP) if \(\delta \in \mathbb{R}_{>0} \) and \(\varepsilon = 0\), and output strictly passive (OSP) if \(\varepsilon \in \mathbb{R}_{>0} \) and \(\delta = 0\) \cite[227-228]{Marquez}.
    \end{definition}
    \vspace{2pt}
    \section{Problem Formulation} \label{sec: problem formulation}
\subsection{Matrix-Gain-Scheduling Architecture}
\begin{figure}[t]
    \vspace{2pt}
    \centering
    \includegraphics{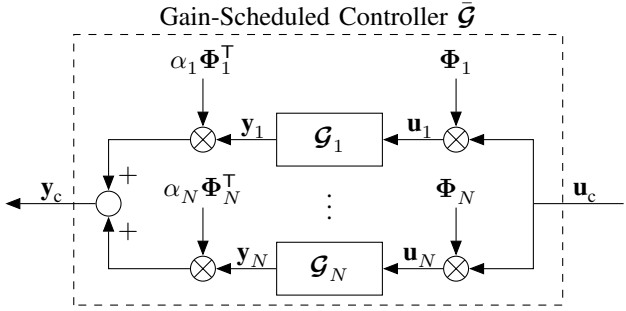}
    \vspace{-5pt}
    \caption{Gain-scheduled controller $\bar{\bm{\mathcal{G}}}$, composed of $N$ parallel VSP subcontrollers. The node $\otimes$ performs matrix multiplication between the scheduling matrices \(\mbs{\Phi}_i(\bm{\zeta }(t), \mbf{x}(t),t)\) and the signals \(\mbf{u}_{\text{c}}(t)\) and \(\mbf{y}_i(t)\) resulting in \cref{eq:GS i/o}. The positive constants $\alpha_i$ are used to scale the gain.}
    \label{fig:GS}
    \vspace{-10pt}
\end{figure}
Consider the gain-scheduled controller, $\bar{\bm{\mathcal{G}}}$, in \Cref{fig:GS}. There are $N$ VSP subcontrollers \(\bm{\mathcal{G}}_1\), \(\bm{\mathcal{G}}_2, \ldots, \bm{\mathcal{G}}_N\) of the form \(\mbf{y}_i(t) = \mleft( \bm{\mathcal{G}}_i \mbf{u}_i \mright) (t)\) satisfying \Cref{def: VSP}, meaning
\begin{equation}
        \langle \mbf{u}_i, \mbf{y}_i \rangle_T
        \geq 
        \beta_i + \delta_i \mleft\| \mbf{u}_i \mright\|^{2}_{2T} + \varepsilon_i \mleft\| \mbf{y}_i \mright\|^{2}_{2T}, \quad \forall T \in \mathbb{R}_{\geq0}, \label{eq: VSP}
\end{equation} 
for \( i \in \mathcal{N} = \{1, \ldots, N \}\) with \(\beta_i \in \mathbb{R}_{\leq 0} \) and \(\delta_i, \varepsilon_i \in \mathbb{R}_{>0} \). The subcontrollers could be linear or nonlinear. The gain-scheduled controller input-output map can be written in terms of the individual subcontroller inputs, outputs, and scheduling matrices as
\begin{subequations} \label{eq:GS i/o}
    \begin{align} 
        \mbf{u}_i(t) &= \mbs{\Phi}_i(\bm{\zeta }(t), \mbf{x}(t),t) \mbf{u}_{\text{c}}(t), \label{eq:GS i/o u_i}
        \\
        \mbf{y}_{\text{c}}(t) &= \sum_{i \in \mathcal{N}} \alpha_i \mbs{\Phi}_{i}^{\trans} (\bm{\zeta }(t), \mbf{x}(t),t) \mbf{y}_{i}^{}(t) \label{eq:GS i/o y_c},
    \end{align}
\end{subequations}
for $\mbf{u}_{\text{c}}, \mbf{y}_{\text{c}} \in \mathbb{R}^{n}$, $\alpha_i \in \mathbb{R}_{>0}$, and $\mbs{\Phi}_{i}^{} \in \mathbb{R}^{n \times n}$ for all $i \in \mathcal{N}$. The variable $\mbs{\zeta}(t)$ represents any external signal convenient for scheduling, while $\mbf{x}(t)$ is the states of the plant under control. For simplicity, the notation $\mbs{\Phi}_{i}^{}(\bm{\zeta }(t), \mbf{x}(t),t)$ is abbreviated to $\mbs{\Phi}_{i}^{}(t)$. 
\subsection{Scheduling Matrix Properties}\label{sub: scheduling matrix properties}
Consider the set of scheduling matrices \(\mbs{\Phi}_{i}(t) \in \mathbb{R}^{n \times n}\) for \(i \in \mathcal{N}\). With abuse of set notation, denote a time dependent set, $\mathcal{F}(t)$, as the index set of all full rank scheduling matrices at time \(t \in \mleft[ 0, T \mright]\) for \(T \in \mathbb{R}_{\geq0}\). That is, for \(t \in \mleft[ 0, T \mright]\)
\begin{equation}
    \mathcal{F}(t) = \mleft\{\,i \in \mathcal{N} \mid \rank \mleft(\mbs{\Phi}_i(t) \mright) = n \,\mright\}.
\end{equation}
\vspace{-10pt}
\begin{definition}[Active Scheduling Matrices]
    For any given gain-scheduled controller of type shown in \Cref{fig:GS} with scheduling matrices $\mbs{\Phi}_{i}^{}(t)\in \mathbb{R}^{n \times n}$ and $i \in \mathcal{N}$, the scheduling matrices are said to be
    \begin{itemize}
        \item \emph{active} if at all times, there exists at least one nonzero scheduling matrix, meaning $\forall t \in \mleft[ 0, T \mright]$, \(\exists i \in \mathcal{N}\) such that \(\mbs{\Phi}_i(t) \neq \mbf{0},\) and
        \item \emph{strongly active} if at all times, there exists at least one full rank scheduling matrix, meaning $\forall t \in \mleft[ 0, T \mright]$, \(\mathcal{F}(t) \neq \varnothing\).
    \end{itemize}
\end{definition}
\vspace{3pt}
\begin{lemma} \label{lemma: eigenvalues }
    Consider the gain-scheduled controller, $\bar{\bm{\mathcal{G}}}$, in \Cref{fig:GS}. Provided the scheduling matrices are strongly active, then
    \begin{equation*}
        \sum_{i \in \mathcal{N}} \lambda_{\min}\mleft( \mbs{\Phi}_{i}^{\trans}(t)\mbs{\Phi}_{i}^{}(t) \mright) 
        =
        \sum_{i \in \mathcal{F}(t)}\nu^{2}_{i}(t)
        >
        0, \quad \forall t \in \mleft[ 0, T \mright],
    \end{equation*}
    where \(\nu_{i}(t)\) is the smallest singular value of \(\mbs{\Phi}_{i}(t)\).
\end{lemma}
\begin{proof}
    Assume \(\exists t \in \mleft[ 0, T \mright]\), for \(T \in \mathbb{R}_{\geq0}\), such that \(\mathcal{F}(t) \neq \varnothing\). Therefore, \(\forall i \in \mathcal{F}(t)\), \(\mbs{\Phi}_{i}(t)\) is full rank, and its smallest singular value, \(\nu_{i}(t) \), is strictly positive. It follows that \(\mbs{\Phi}_{i}^{\trans}(t)\mbs{\Phi}_{i}^{}(t)\) is also full rank and its minimum eigenvalue is exactly \(\nu_{i}^{2}(t) \). Additionally, $\forall j \in \mathcal{N} \setminus \mathcal{F}(t)$, \(\mbs{\Phi}_{j}(t)\) is rank deficient, therefore, \(\sum_{j \in \mathcal{N} \setminus \mathcal{F}(t)}\lambda_{\min} \mleft(\mbs{\Phi}_{j}^{\trans}(t)\mbs{\Phi}_{j}^{}(t) \mright) = 0 \).
    \hspace*{\fill}~\QED\par\endtrivlist\unskip
\end{proof}

For the remainder of this paper, it is assumed that \(\alpha_i \in \mathbb{R}_{>0}\) for all \(i \in \mathcal{N}\). Moreover, the scheduling matrices are assumed to be bounded in the sense that
\begin{equation}\label{eqn: boundedness}
    \sup_{t \in \mleft[ 0, T \mright] }\mleft\| \mbs{\Phi}_{i}(t) \mright\|_{2}^{2}
    = 
    \sup_{t \in \left[ 0, T \right]} \sigma_{i}^2(t) < \infty,
\end{equation}
for all $i \in \mathcal{N}$ and \(T \in \mathbb{R}_{\geq0}\), where \(\sigma_{i}(t)\) is the largest singular value of $\mbs{\Phi}_{i}(t)$. 
    \section{Main Contribution} \label{sec: main results}
The main result of this paper, which is presented in this section, is showing that the gain-scheduled controller in \Cref{fig:GS} is VSP when the subcontrollers are VSP and the scheduling matrices are strongly active and bounded. Doing so is a generalization of \cite{Damaren_passive_map, Forbes_Damaren} from scheduling signals to scheduling matrices. 
\vspace{2pt}
\subsection[Passivity properties of gain-scheduled controller G]{Passivity properties of gain-scheduled controller $\bar{\bm{\mathcal{G}}}$}
Using the input-output map of the gain-scheduled controller $\bar{\bm{\mathcal{G}}}$ in \cref{eq:GS i/o}, it follows that
\begin{align} 
    &\langle \mbf{u}_{\text{c}}, \bar{\bm{\mathcal{G}}} \mbf{u}_{\text{c}}^{}\rangle_T
     = 
     \langle \mbf{u}_{\text{c}}, \mbf{y}_{\text{c}}^{}\rangle_T
     =
        \int_{0}^{T} \mbf{u}^{\trans}_{\text{c}}(t) \mleft( \sum_{i \in \mathcal{N}} \alpha_i \mbs{\Phi}_{i}^{\trans}(t) \mbf{y}_{i}^{}(t) \mright)\, \dt \nonumber
     \\%
     &=
     \sum_{i \in \mathcal{N}} \int_{0}^{T} \alpha_i \mbf{u}^{\trans}_{\text{c}}(t) \mbs{\Phi}_{i}^{\trans}(t)\mbf{y}_{i}^{}(t)\, \dt 
     =
        \sum_{i \in \mathcal{N}} \alpha_i \left\langle \mbf{u}_i, \mbf{y}_i \right\rangle_{T}. \label{eq: expanded u.Ty}
 \end{align} 
% -------------------------------------------------------------------------------------- %
% ---------------------------------- ISP Proof ----------------------------------------- %
% -------------------------------------------------------------------------------------- %
\begin{theorem} \label{thm: ISP connection is ISP}
    The gain-scheduled controller $\bar{\bm{\mathcal{G}}}$ in \Cref{fig:GS} is ISP if each subcontroller $\bm{\mathcal{G}}_i$ is ISP and the scheduling matrices are strongly active. 
\end{theorem}
\begin{proof}
    Substituting into \cref{eq: expanded u.Ty} the ISP version of \cref{eq: VSP} provided in \Cref{def: VSP}, it follows that
    \begin{align} \label{eqn: ISP before Rayleigh}
        \langle \mbf{u}_{\text{c}}, \bar{\bm{\mathcal{G}}} \mbf{u}_{\text{c}}^{}\rangle_T
        =
        \sum_{i \in \mathcal{N}} \alpha_i \left\langle \mbf{u}_i, \mbf{y}_i \right\rangle_{T}
        &\geq
            \sum_{i \in \mathcal{N}} \alpha_i \mleft( \hat{\beta}_i + \delta_i \mleft\| \mbf{u}_i \mright\|^{2}_{2T} \mright) \nonumber
            \\%
        &\geq
            \hat{\beta}
            +
            \delta_{\min}\sum_{i \in \mathcal{N}}\mleft\| \mbf{u}_i \mright\|^{2}_{2T},
        % &=
        %     \hat{\beta}
        %     +
        %     \delta_{\min} \sum_{i \in \mathcal{N}}\mleft\| \mbs{\Phi}_{i}\mbf{u}_{\text{c}} \mright\|^{2}_{2T},
    \end{align}
    with \(\hat{\beta}_i \in \mathbb{R}_{\leq 0} \) and \(\delta_i \in \mathbb{R}_{>0} \) for all $i \in \mathcal{N}$, and
    \begin{align} \label{eqn: ISP constants}
        \hat{\beta} &= \sum_{i \in \mathcal{N}} \alpha_i \hat{\beta}_i \leq 0,& \delta_{\min} &= \min_{i \in \mathcal{N}} \alpha_i \delta_i > 0.
    \end{align}
    % Applying \Cref{lemma: Rayleigh Inequality} to \cref{eqn: ISP before Rayleigh} and defining \(\lambda_{\min}^{(i)}(t)\) as the minimum eigenvalue of \(\mbs{\Phi}_{i}^{\trans}(t)\mbs{\Phi}_{i}(t)\) leads to
    Substituting \cref{eq:GS i/o u_i} into \cref{eqn: ISP before Rayleigh} and applying the Rayleigh inequality leads to
    \begin{align} \label{eqn: ISP after Rayleigh}
        \langle \mbf{u}_{\text{c}}, \bar{\bm{\mathcal{G}}} \mbf{u}_{\text{c}}^{}\rangle_T
        &\geq
            \hat{\beta }
            +
            \delta_{\min}
            \int_{0}^{T} \sum_{i \in \mathcal{N}}\lambda_{\min}^{(i)}(t)\mleft\| \mbf{u}_{\text{c}}(t)\mright\|_{2}^{2}\, \dt, 
    \end{align}
    where \(\lambda_{\min}^{(i)}(t)\) is the minimum eigenvalue of \(\mbs{\Phi}_{i}^{\trans}(t)\mbs{\Phi}_{i}(t)\). For \(T \in \mathbb{R}_{\geq0}\), provided the scheduling matrices are strongly active, that is, \(\mathcal{F}(t) \neq \varnothing\), $\forall t \in \mleft[ 0, T \mright] $, \Cref{lemma: eigenvalues } can be applied to \cref{eqn: ISP after Rayleigh} by defining \(\nu_{i}(t)\) as the smallest singular value of \(\mbs{\Phi}_{i}(t)\). This results in
    \begin{align} \label{eqn: ISP final result}
        \langle \mbf{u}_{\text{c}}, \bar{\bm{\mathcal{G}}} \mbf{u}_{\text{c}}^{}\rangle_T
        &\geq 
            \hat{\beta }
            +
            \delta_{\min} \int_{0}^{T} \sum_{i \in \mathcal{F}(t)} \nu^{2}_{i}(t)\mleft\| \mbf{u}_{\text{c}}(t)\mright\|_{2}^{2}\, \dt \nonumber
        \\%
        &\geq
            \hat{\beta }
            +
            \delta_{\min} \nu_{\inf} \mleft\| \mbf{u}_{\text{c}}\mright\|_{2T}^{2} =
            \hat{\beta } 
            +
            \hat{\delta} \mleft\| \mbf{u}_{\text{c}} \mright\|^{2}_{2T},
    \end{align}
    with 
    \begin{align} \label{eqn: ISP final result constants}
        \nu_{\inf}    &= \inf_{t \in \left[ 0, T \right]} \sum_{i \in \mathcal{F}(t)} \nu^{2}_{i}(t) > 0
        ,&
        \hat{\delta}  &= \delta_{\min} \nu_{\inf} > 0.
    \end{align}
    \hspace*{\fill}~\QED\par\endtrivlist\unskip
\end{proof}
% -------------------------------------------------------------------------------------- %
% ---------------------------------- OSP Proof ----------------------------------------- %
% -------------------------------------------------------------------------------------- %
\begin{theorem} \label{thm: OSP connection is OSP}
    The gain-scheduled controller $\bar{\bm{\mathcal{G}}}$ in \Cref{fig:GS} is OSP if each subcontroller $\bm{\mathcal{G}}_i$ is OSP and the scheduling matrices are active. 
\end{theorem}
\begin{proof}
    By defining the augmented matrices 
    \begin{subequations}
        \begin{gather}
            \begin{align}
                \bm{\Psi}(t) &= 
                \begin{bmatrix}
                    \mbs{\Phi}_{1}(t) \\%
                    \vdots       \\%
                    \mbs{\Phi}_{N}(t)\\%
                \end{bmatrix}^{\trans}
                ,&
                \bm{\upsilon}(t) &= 
                \begin{bmatrix}
                    \mbf{y}_1(t) \\%
                    \vdots       \\%
                    \mbf{y}_N(t) \\%
                \end{bmatrix},
            \end{align}
            \\
            \begin{align}
                \bm{\Lambda} =
                \diag \mleft(\alpha_1 \mbf{1}, \ldots,  \alpha_N \mbf{1}\mright),
            \end{align}
        \end{gather}
    \end{subequations}
    it follows that \cref{eq:GS i/o y_c} can be written as
    \(
        \mbf{y}_{\text{c}}(t) = \bm{\Psi}(t) \bm{\Lambda} \bm{\upsilon}(t)
    \).
    Using the Rayleigh inequality twice, it follows that
    \begin{equation} \label{eqn: yc to yi}
        \mleft\|\mbf{y}_{\text{c}}(t) \mright\|_{2}^{2}
        =
        \mleft\|\bm{\Psi}(t) \bm{\Lambda} \bm{\upsilon}(t) \mright\|_{2}^{2}
        \leq 
        \alpha_{\max}^{2}\sigma_{\mbs{\Psi}}^{2}(t) \mleft\|\bm{\upsilon}(t) \mright\|_{2}^{2},
    \end{equation}
    where \(\sigma_{\mbs{\Psi}}(t)\) is the largest singular value of \(\bm{\Psi}(t)\) and \(\alpha_{\max} = \max_{i \in \mathcal{N}} \alpha_i\). 
    Provided the scheduling matrices are active, then \(\sigma_{\mbs{\Psi}}(t)\in \mathbb{R}_{>0}\). Rearranging \cref{eqn: yc to yi} yields
    \begin{align} \label{eqn: yc to yi after Rayleigh}
        \frac{1}{\alpha_{\max}^{2}\sigma_{\mbs{\Psi}}^{2}(t)} \mleft\| \mbf{y}_{\text{c}}(t)\mright\|_{2}^{2}
        \leq 
        \mleft\|\mbs{\upsilon}(t)\mright\|_{2}^{2} 
        = \sum_{i \in \mathcal{N}} \mleft\| \mbf{y}_{i}(t) \mright\|_{2}^{2}. 
    \end{align}
    Substituting into \cref{eq: expanded u.Ty} the OSP version of \cref{eq: VSP} provided in \Cref{def: VSP}, it follows that 
    \begin{align} \label{eqn: OSP before Rayleigh}
        \hspace{-5pt}
        \langle \mbf{u}_{\text{c}}, \bar{\bm{\mathcal{G}}} \mbf{u}_{\text{c}}^{}\rangle_T
        =
        \sum_{i \in \mathcal{N}} \alpha_i \left\langle \mbf{u}_i, \mbf{y}_i \right\rangle_{T}
        &\geq
            \sum_{i \in \mathcal{N}} \alpha_i \mleft(\bar{\beta}_i + \varepsilon_i\mleft\| \mbf{y}_{i} \mright\|^{2}_{2T} \mright) \nonumber
        \\%
        &\geq
            \bar{\beta } 
            +
            \varepsilon_{\min}\sum_{i \in \mathcal{N}}  \mleft\| \mbf{y}_{i} \mright\|^{2}_{2T},
    \end{align}
    with \(\bar{\beta}_i \in \mathbb{R}_{\leq 0} \) and \(\varepsilon_i \in \mathbb{R}_{>0} \) for all $i \in \mathcal{N}$, and
    \begin{align}\label{eqn: OSP constants}
        \bar{\beta}         &= \sum_{i \in \mathcal{N}} \alpha_i \bar{\beta}_i \leq 0,&
        \varepsilon_{\min}  &= \min_{i \in \mathcal{N}} \alpha_i \varepsilon_i > 0.
    \end{align}
    Substituting \cref{eqn: yc to yi after Rayleigh} into \cref{eqn: OSP before Rayleigh} leads to
    \begin{align}\label{eqn: OSP final result}
        \langle \mbf{u}_{\text{c}}, \bar{\bm{\mathcal{G}}} \mbf{u}_{\text{c}}^{}\rangle_T
        &\geq
            \bar{\beta} 
            +
            \varepsilon_{\min} \int_{0}^{T} \frac{1}{\alpha_{\max}^{2}\sigma_{\mbs{\Psi}}^{2}(t)} \mleft\| \mbf{y}_{\text{c}} (t)\mright\|^{2}_{2}\, \dt\nonumber
        \\%
        &\geq
            \bar{\beta } 
            +
            \frac{\varepsilon_{\min}}{\alpha_{\max}^{2}\bar{\sigma}_{\mbs{\Psi}}^{2}} \mleft\| \mbf{y}_{\text{c}}\mright\|^{2}_{2T} =
            \bar{\beta } 
            +
            \bar{\varepsilon} \mleft\| \mbf{y}_{\text{c}} \mright\|^{2}_{2T},
    \end{align}
    with 
    \begin{align} \label{eqn: OSP final result constants}
        \bar{\sigma}_{\mbs{\Psi}} &= \sup_{t \in \left[ 0, T \right]} \sigma_{\mbs{\Psi}}(t) > 0,
        &
        \bar{\varepsilon} = \frac{\varepsilon_{\min}}{\alpha_{\max}^{2}\bar{\sigma}_{\mbs{\Psi}}^{2}} > 0.
    \end{align}
    \hspace*{\fill}~\QED\par\endtrivlist\unskip
\end{proof}
% --------------------------------------------------------------------------- %
% --------------------------------------------------------------------------- %
\subsection{Discussion}
Given that the $N$ VSP subcontrollers \(\bm{\mathcal{G}}_1, \bm{\mathcal{G}}_2, \ldots, \bm{\mathcal{G}}_N\) in \Cref{fig:GS} are assumed to be VSP, they are also ISP and OSP simultaneously. Consider the gain-scheduled controller $\bar{\bm{\mathcal{G}}}$ in \Cref{fig:GS}. The condition required for  $\bar{\bm{\mathcal{G}}}$ to be ISP, as stated in \Cref{thm: ISP connection is ISP}, is more restrictive than the OSP case in \Cref{thm: OSP connection is OSP}, since the existence of a full rank scheduling matrix at all times also implies the existence of a nonzero scheduling matrix at all times. Therefore, the matrix-gain-scheduling of $N$ VSP subcontrollers as per \Cref{fig:GS} satisfies \Cref{thm: ISP connection is ISP} and \Cref{thm: OSP connection is OSP} simultaneously provided the scheduling matrices are strongly active. Consequently, combining \cref{eqn: ISP final result} and \cref{eqn: OSP final result} provides
\begin{equation*}
    \langle \mbf{u}_{\text{c}}, \bar{\bm{\mathcal{G}}} \mbf{u}_{\text{c}}^{}\rangle_T
    \geq
    \frac{\hat{\beta} + \bar{\beta}}{2}
    +
    \frac{\hat{\delta}}{2} \mleft\| \mbf{u}_{\text{c}} \mright\|^{2}_{2T}
    +
    \frac{\bar{\varepsilon}}{2} \mleft\| \mbf{y}_{\text{c}} \mright\|^{2}_{2T},
\end{equation*}
with \(\hat{\beta}\), \(\hat{\delta} \), \(\bar{\beta}\), and \(\bar{\varepsilon} \) defined in \cref{eqn: ISP constants}, \cref{eqn: ISP final result constants}, \cref{eqn: OSP constants}, and \cref{eqn: OSP final result constants}, respectively.

As required in \Cref{thm: ISP connection is ISP}, at all times, there must be at least one full rank scheduling matrix for the gain-scheduled controller \(\bar{\bm{\mathcal{G}}}\) to be ISP. To elaborate, assume at time \(t \in \mleft[ 0, T \mright] \), with \(T \in \mathbb{R}_{\geq0}\), the scheduling matrix \(\mbs{\Phi}_j(t)\) is full rank for some \(j \in \mathcal{F}(t)\). Then, \(\alpha_j \mbs{\Phi}(t)_{j}^{\trans} \mbf{y}_{j}(t)\) and \(\mbf{u}_j(t)\) are both nonzero, provided that \(\mbf{u}_{\text{c}}(t)\) and \(\mbf{y}_{j}(t)\) are nonzero.
This can be thought of as a direct extension of \cite[Theorem 1]{Damaren_passive_map}, where it is required for at least one scheduling signal to be nonzero at all times. Additionally, the gain-scheduled ISP coefficient in \cite{Damaren_passive_map, Forbes_Damaren} is a special case of \(\hat{\delta}\) in \cref{eqn: ISP final result constants}, where for a scalar scheduling signal \(s_i(t) \in \mathbb{R}\), its smallest singular value is \(|s_i(t)|\), and with \(\alpha_i = 1\), \cref{eqn: ISP final result constants} leads to \(\hat{\delta} =  \inf_{t \in \mleft[ 0, T \mright]} \sum_{i \in \mathcal{F}(t)} s_{i}^{2}(t) \delta_{\min}\), with \(\delta_{\min} = \min_{i \in \mathcal{N}} \delta_i\). 

In \cite[Theorem 5.2]{Forbes_Damaren}, to show that a gain-scheduled controller composed of a family of VSP controllers possesses finite gain, the scalar scheduling signals are required to be bounded as \(s_i(t)\in \mathcal{L}_{\infty}\). Similarly, in \Cref{thm: OSP connection is OSP}, \(\sigma_{\mbs{\Psi}}^{2}(t)\) is required to be nonzero and finite for all \(t \in \mleft[ 0, T \mright]\). The scheduling matrices being active guarantees \(\sigma_{\mbs{\Psi}}^{2}(t) \in \mathbb{R}_{>0}\). Moreover, for the symmetric positive semi-definite matrix \(\bm{\Psi}^{\trans} (t) \bm{\Psi}(t)\) it follows that
\begin{align*}
    \sigma_{\mbs{\Psi}}^{2}(t)
    &\leq \trace{\mleft( \bm{\Psi}^{\trans} (t) \bm{\Psi}(t) \mright) }
    = 
        \sum_{i \in \mathcal{N}} \trace{\mleft(\mbs{\Phi}_{i}^{\trans} (t)\mbs{\Phi}_{i}(t)\mright)}
    \\%
    &\leq 
        n \sum_{i \in \mathcal{N}} \lambda_{\max}\mleft( \mbs{\Phi}_{i}^{\trans} (t)\mbs{\Phi}_{i}(t) \mright)
    = 
        n \sum_{i \in \mathcal{N}} \left\| \mbs{\Phi}_i(t) \right\|_{2}^{2}.
\end{align*}
Since the scheduling matrices are assumed to be bounded as per \cref{eqn: boundedness}, it follows that 
\begin{align*}
    \bar{\sigma}_{\mbs{\Psi}}^{2} 
    = 
    \sup_{t \in \left[ 0, T \right]} \sigma_{\mbs{\Psi}}^{2}(t)
    % &\leq
    % \sup_{t \in \mleft[ 0, T \mright] } n \sum_{i \in \mathcal{N}} \mleft\| \mbs{\Phi}_{i}(t) \mright\|_{2}^{2}
    % \\%
    &\leq
    n \sum_{i \in \mathcal{N}} \sup_{t \in \mleft[ 0, T \mright]}\mleft\| \mbs{\Phi}_{i}(t) \mright\|_{2}^{2} < \infty.
\end{align*}

Additionally, as discussed in \cite[Proposition 2.11]{Sepulchre}, the input-output modification of the subcontrollers described in \cref{eq:GS i/o} does not violate passivity. This input-output modification is being used in a novel way to gain-schedule subcontrollers using scheduling matrices. 
    \section{Application Example} \label{sec: simulation}
\begin{figure}[t]
    \vspace{2pt}
    \centering
    \includegraphics{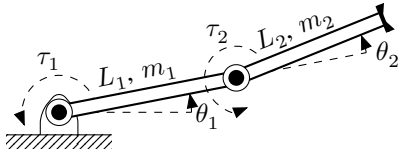}
    \vspace{-5pt}
    \caption{Rigid two-link robotic manipulator with joint angles $\theta_1$ and $\theta_2$ and joint torques $\tau_1$ and $\tau_2$.}
    \label{fig:robot}
\end{figure}
\begin{table}[t]
    \vspace{3pt}
    \caption{Two-Link Manipulator Properties}
    \vspace{-7pt}
    \label{tab: Manipulator Properties}
    \small
    \begin{tabularx}{\linewidth}{l>{\centering\arraybackslash}X>{\centering\arraybackslash}X}
        \hline
        \hline
        \addlinespace[2pt]
        {\small Link Parameters}
        & 
        {\small Link 1}
        &
        {\small Link 2} 
        \\ 
        \hline
        \addlinespace[2pt]
        Length {$\si{[\meter]}$}                          & {\(\,L_1       = 1.10\)} & {\(\,L_2       = 0.85\)} \\
        Measured Length {$\si{[\meter]}$}                 & {\(\,\bar{L}_1 = 1.08\)} & {\(\,\bar{L}_2 = 0.83\)} \\
        Mass {$\si{[\kilo\gram]}$}                        & {\(m_1       = 0.40\)}   & {\(m_2       = 0.90\)} \\
        Measured Mass {$\si{[\kilo\gram]}$}               & {\(\bar{m}_1 = 0.44\)}   & {\(\bar{m}_2 = 0.99\)} \\
        \hline
        \hline
    \end{tabularx}
    \vspace{-9pt}
\end{table}
\subsection{Plant Description}
To demonstrate the benefits of scheduling matrices within the framework of \Cref{fig:GS}, the control of a rigid two-link robotic manipulator is considered, as shown in \Cref{fig:robot}. For simplicity, the robot is assumed to be planar, with a fixed base, and no forces acting on the end-effector \cite[177-181]{Craig}. The equations of motion of the two-link robot are given by
\begin{equation} \label{eq:robot_dynamics}
    \mbf{M}(\mbf{q}(t))\ddot{\mbf{q}}(t) = 
    \mbf{f}_{\text{non}}(\mbf{q}(t), \dot{\mbf{q}}(t))
    +
    \mbf{u}(t), 
\end{equation}
where \(\mbf{M}(\mbf{q}(t)) = \mbf{M}^{\trans}(\mbf{q}(t)) \succ 0\) is the mass matrix, \(\mbf{f}_{\text{non}}(\mbf{q}(t), \dot{\mbf{q}}(t))\) captures the nonlinear inertial and Coriolis forces, \(\mbf{u}(t) = \begin{bmatrix} \tau_1(t) & \tau_2(t) \end{bmatrix}^{\trans}\) are the joint torques, and \(\mbf{q}(t) = \begin{bmatrix} \theta_1(t) & \theta_2(t) \end{bmatrix}^{\trans}\) are the generalized coordinates. The passive map associated with \cref{eq:robot_dynamics} is joint torques to joint rates, which is \(\mbf{u}(t) \to \dot{\mbf{q}}(t)\).
\subsection{Trajectory}
\begin{table}[t]
    \centering
    \vspace{4pt}
    \caption{Discrete Joint Angles for Trajectory Generation}
    \vspace{-7pt}
    \label{tab:trajectory points}
    \small
    \begin{tabularx}{\linewidth}{*{2}{>{\centering\arraybackslash}X}}
        \hline
        \hline
        \addlinespace[2pt]
        {\small Discrete Time Points \(t_k\)} 
        & 
        {\small Desired Joint Angle \( \bm{\theta}_\text{d}(t_{k})\)} 
        \\
        {\small $\si{[s]}$} 
        & 
        {\small $\si{[deg]}$} 
        \\
        \addlinespace[1pt]
        \hline
        \addlinespace[2pt]
        \(t_0 = 0.0,\,t_1 = 0.5 \) & \(\begin{bmatrix} -90^\circ & 150^\circ \end{bmatrix}^{\trans} \) \\
        \(t_2 = 1.0,\,t_3 = 2.0\) & \(\begin{bmatrix} -60^\circ & \phantom{1}90^\circ \end{bmatrix}^{\trans} \) \\
        \(t_4 = 3.0 \) & \(\begin{bmatrix} \phantom{-}45^\circ  & \phantom{1}60^\circ \end{bmatrix}^{\trans} \) \\
        \(t_5 = 5.0 \) & \(\begin{bmatrix} \phantom{-}60^\circ  & \phantom{1}45^\circ \end{bmatrix}^{\trans} \) \\
        \(t_6 = 6.0,\,t_7 = 6.5 \) & \(\begin{bmatrix} \phantom{1}90^\circ  & -60^\circ \end{bmatrix}^{\trans} \) \\
        \(t_8 = 7.5,\,t_9 = 8.5 \) & \(\begin{bmatrix} 150^\circ & -90^\circ \end{bmatrix}^{\trans} \) \\
        \addlinespace[2pt]
        \hline
        \hline
    \end{tabularx}
    \vspace{-10pt}
\end{table}
The control objective is to have the two-link robot track a position and rate trajectory. The position trajectory is \(\bm{\theta}_\text{d}(t) = \begin{bmatrix} \theta_{\text{d}, 1}(t) & \theta_{\text{d}, 2}(t) \end{bmatrix}^{\trans}\), and the rate trajectory is  \(\dot{\bm{\theta}}_\text{d}(t)\). This is achieved by choosing discrete joint angles \(\bm{\theta}_\text{d}(t_k)\) and \(\bm{\theta}_\text{d}(t_{k+1})\) at times \(t_k\) and \(t_{k+1}\), and interpolating between them as such
\begin{subequations}\label{eqn:trajectory_interpolation}
    \begin{gather}
        \begin{align}
            \eta(t)
            &= \frac{t - t_k}{t_{k+1} - t_k},&
            p_5(t) 
            &= 6 \eta^5 - 15\eta^4 + 10\eta^3,
        \end{align}
        \\%
        \bm{\theta}_\text{d}(t) 
        = 
        p_5(t) \mleft(\bm{\theta}_\text{d}(t_{k+1}) - \bm{\theta}_\text{d}(t_{k})\mright) 
            +
            \bm{\theta}_\text{d}(t_{k}).
    \end{gather}
\end{subequations}
As shown in \Cref{tab:trajectory points}, the desired discrete joint angles are chosen such that the joint angles operate within \(\left[-90^\circ, 150^\circ\right]\).
\subsection{SPR Control Synthesis}
The subcontrollers to be gain scheduled will be SPR controllers with feedthrough, which are in turn VSP \cite{Marquez}. There are many ways to synthesize SPR controllers. In \cite{Forbes_Damaren}, SPR control synthesis is achieved through construction of SPR transfer functions for a given Hurwitz polynomial \cite{Damaren}. Others solve a convex optimization problem subject to linear matrix inequalities (LMIs) to synthesize SPR controllers \cite{Gapski, Takashi, Forbes_Dual, Forbes_dilated}. Herein, the Kalman-Yakubovich-Popov (KYP) lemma \cite{KYP} and gain matrix \(\mbf{K}\) from the linear-quadratic regulator (LQR) problem are used to synthesize the SPR controllers, as is similarly done in \cite{Benhabib,Damaren_passive_map, walsh}. 

The LQR problem requires a linearized version of \cref{eq:robot_dynamics}. Given the form of the mass matrix shown in \cite[180]{Craig}, the nonlinearity of $\mbf{M}(\mbf{q}(t))$ comes from the $\cos(\theta_2)$ term. Therefore, to cover the range of possible joint angles during the desired trajectory, three linearization points are chosen with \(\theta_{i,2} \in \{150^\circ, 60^\circ, -90^\circ\}\). Since the SPR controller is a rate-based controller, a proportional control prewrap is then added to the system to control the joint displacements of the system. This prewrap does not violate the passive map of the system \cite{Damaren_passive_map}. The linearization of the prewrapped model about \(\bar{\mbf{q}}_{i} = \begin{bmatrix} 0 & \theta_{i,2} \end{bmatrix}^{\trans}\) is given by
\begin{align} \label{eq:linearized_model}
        \delta \dot{\mbf{x}}(t)
        &= 
        \mbf{A}_{i} \delta \mbf{x}(t)
            +
        \mbf{B}_{i}
        \delta \mbf{u}(t),
        &
        \delta \mbf{y}(t) &= \mbf{C}_{i} \delta \mbf{x}(t),
\end{align}
with \(\delta \mbf{x}(t) = \begin{bmatrix}
    \delta \mbf{q}(t) & \delta \dot{\mbf{q}}(t)
\end{bmatrix}^{\trans}\) and 
% \begin{subequations} \label{eq:linearized_model_A_B_C}
    \begin{align*}
        \mbf{A}_{i}
        &= 
            \begin{bmatrix}
                \mbf{0}                               & \eye \\%
                -\bar{\mbf{M}}^{-1}(\bar{\mbf{q}}_{i})\mbf{K}_{\text{p}} & \mbf{0}
            \end{bmatrix}, &
            \hspace{-4pt}
            \mbf{B}_{i} 
            &=
            \begin{bmatrix}
                \mbf{0}  \\%
                \bar{\mbf{M}}^{-1}(\bar{\mbf{q}}_{i})
            \end{bmatrix}, & 
            \hspace{-4pt}
            \mbf{C}_{i}
            &= 
            \begin{bmatrix} 
                \mbf{0} \\
                \eye 
            \end{bmatrix}^{\trans}\hspace{-3pt},
    \end{align*}
% \end{subequations}
where \(\bar{\mbf{M}}\mleft( \bar{\mbf{q}}_i \mright) \) is the measured mass matrix constructed using the measured link lengths and masses in \Cref{tab: Manipulator Properties} and \(\mbf{K}_{\text{p}}\) is the proportional gain matrix in \Cref{tab: Controller Properties}.
\begin{table}[t]
    \vspace{4pt}
    \caption{Controller Design Parameters}
    \vspace{-7pt}
    \label{tab: Controller Properties}
    \small
    \begin{tabularx}{\linewidth}{l>{\centering\arraybackslash}p{0.17\linewidth}>{\centering\arraybackslash}p{0.445\linewidth}}
        \hline
        \hline
        \addlinespace[2pt]
        {\small Properties} &
        {\small Symbol} &
        {\small Value} \\ 
        \hline
        \addlinespace[2pt]
        Proportional Gain               & {\(\mbf{K}_{\text{p}}\)} & {\(\diag \mleft(35, 35 \mright)\)} \\
        \addlinespace[2pt]
        \hline
        \addlinespace[2pt]
        \multirow{2}{*}{LQR Weights} & {\(\mbf{Q}_\text{LQR}\)} & {\(\diag \mleft(0.33, 0.25, 180, 180\mright)^{-2}\)} \\
                                     & {\(\mbf{R}_\text{LQR}\)} & {\(\diag \mleft(15, 15 \mright)^{-2}\)} \\
        \addlinespace[2pt]
        \hline
        \addlinespace[2pt]
        Feedthrough & {\(\delta\)} & {0.0001} \\
        \hline
        \hline
    \end{tabularx}
    \vspace{-10pt}
\end{table}
Additionally, the LQR problem's state and input weight matrices, \(\mbf{Q}_\text{LQR}\) and \(\mbf{R}_\text{LQR}\), are tabulated in \Cref{tab: Controller Properties} following Bryson's rule \cite{brysons}.

Using the linearized model in \cref{eq:linearized_model} with the LQR state and input weight matrices in \Cref{tab: Controller Properties}, the gain matrix \(\mbf{K}_{i}\) is computed for each linearization point by solving the algebraic Riccati equation (ARE) \cite{brysons}. The SPR control synthesis is then completed similar to \cite{Damaren_passive_map, Benhabib, walsh} by using the KYP lemma to set
\begin{align*}
    \mbf{A}_{\text{c}, i} &= \mbf{A}_{i} - \mbf{B}_{i}\mbf{K}_{i}, &
    \mbf{C}_{\text{c}, i} &= \mbf{K}_{i}, &
    \mbf{B}_{\text{c}, i} &=\mbf{P}_{i}^{-1} \mbf{K}_{i}^{\trans}, 
\end{align*}
where \(\mbf{P}_{i} = \mbf{P}_{i} ^{\trans} \succ 0 \) is the solution to the Lyapunov equation,
\(
    \mbf{A}_{\text{c}, i}^{\trans}\mbf{P}_{i} + \mbf{P}_{i}\mbf{A}_{\text{c}, i} = -\mbf{Q}_{i}
\), for \(\mbf{Q}_{i}  = \mbf{Q}_{i}^{\trans} \succ 0\).

Finally, an SPR controller by itself is not VSP. However, an SPR controller in a parallel feedforward connection with an arbitrary constant gain \(\delta \in \mathbb{R}_{>0}\) is VSP \cite{Marquez}. Therefore, for each linearization point \(\bar{\mbf{q}}_{i}\), a VSP controller, \(\bm{\mathcal{G}}_i : \mathcal{L}_{2e} \to \mathcal{L}_{2e}\), can be synthesized with the state-space form
\begin{align*}
    \dot{\mbf{x}}_i(t) &= \mbf{A}_{\text{c}, i}  \mbf{x}_i(t) + \mbf{B}_{\text{c}, i}  \mbf{u}_i(t) ,
    &
    \mbf{y}_i(t) &= \mbf{C}_{\text{c}, i}  \mbf{x}_i(t) + \mbf{D}_{\text{c}}  \mbf{u}_i(t) ,
\end{align*}
where \(\mbf{D}_{\text{c}} = \delta \eye\), with \(\delta\) in \Cref{tab: Controller Properties}.
\subsection{Scheduling Signals} \label{subsec: scheduling_signals}
Historically, gain-scheduled controllers have used linear scheduling signals. However, herein, fourth degree polynomials are used as scheduling signals within the scheduling matrices. For the three linearization points \(\bar{\mbf{q}}_{1}\), \(\bar{\mbf{q}}_{2}\), and \(\bar{\mbf{q}}_{3}\), the scheduling signals in \Cref{fig:signals} are defined as
\begin{subequations} \label{eq:scalar_scheduling_signals}
    \begin{align}
        s_1(t) &= 
        \begin{cases}
            1 - \left(\frac{t}{3}\right)^4 & \qquad \phantom{--} 0.0 \leq t \leq 3.0, \\
            0                              & \qquad \phantom{--} 3.0 < t,
        \end{cases} \\
        s_2(t) &= 
        \begin{cases}
            1 - \left(\frac{t-3}{2.8}\right)^4 &\qquad \phantom{-}  0.2 \leq t \leq 5.8 ,\\
            0                                  &\qquad \phantom{-}\text{otherwise},
        \end{cases} \\
        s_3(t) &= 
        \begin{cases}
            0                                    & \qquad  0.0 \leq t < 5.0, \\
            1 - \left(\frac{t-7.5}{2.5}\right)^4 & \qquad  5.0 \leq t \leq 7.0, \\
            1                                    & \qquad  7.0 < t.
        \end{cases}
    \end{align}
\end{subequations}
\begin{figure}[t]
    \centering
    \vspace{5pt}
    \includegraphics[width=0.485\textwidth]{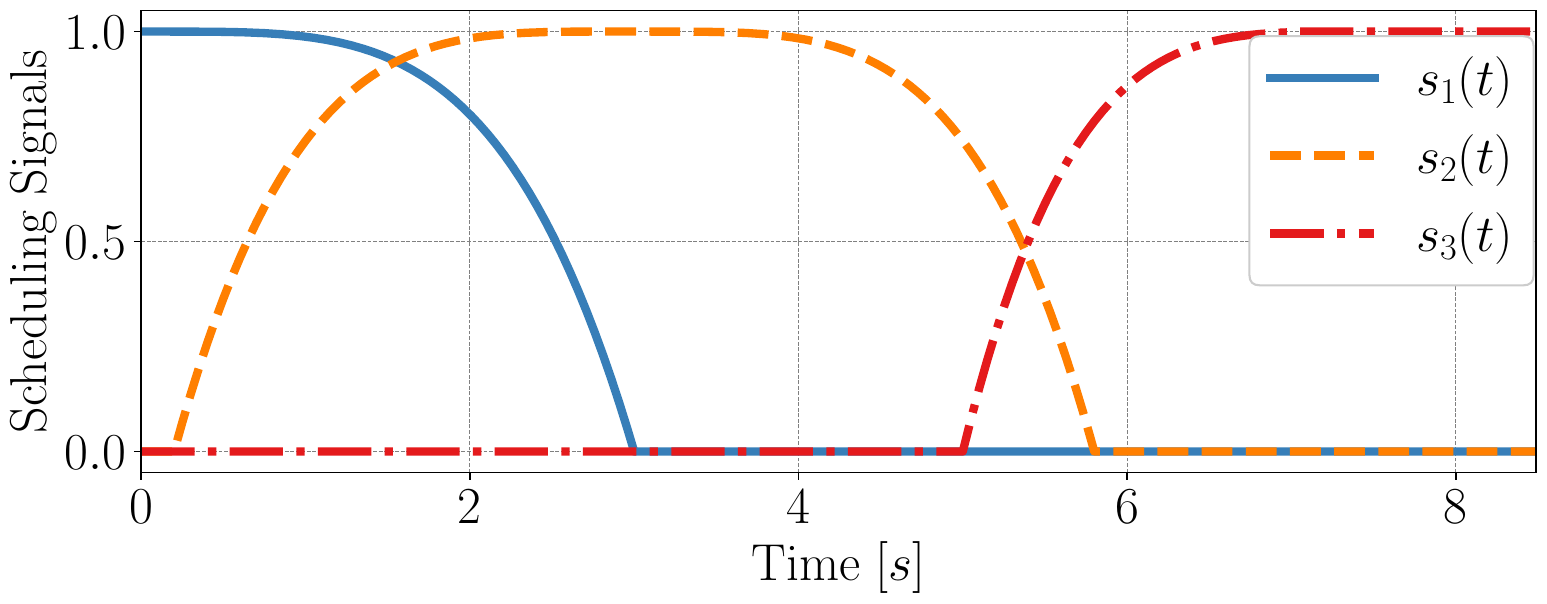}
    \vspace{-20pt}
    \caption{Scalar scheduling signals $s_1(t)$, $s_2(t)$, and $s_3(t)$ defined in \cref{eq:scalar_scheduling_signals}.}
    \label{fig:signals}
    \vspace{-10pt}
\end{figure}
Note that after $T = 7 < t$, $s_3 = 1$, while the other signals are zero. Additionally, all scheduling signals are bounded, and at all times, at least one scheduling signal is active, meaning \(s_i \in \mathcal{L}_{2e}\), \(\forall i \in \mathcal{N} = \{1, 2, 3\}\), and \(\sum_{i \in \mathcal{N}} s_i(t) > 0\) \(\forall t \in \mathbb{R}_{\geq 0}\). Therefore, as required in \cite{Forbes_Damaren}, $s_1$, $s_2$, and $s_3$ are valid scalar signals to preserve the VSP property of the gain-scheduled subcontrollers. 

As per \Cref{fig:GS}, for \(\mbf{u} : \mathbb{R}_{\geq 0} \to \mathbb{R}^{2}\), the scheduling of each subcontroller \(\bm{\mathcal{G}}_i\) requires five hyperparameters: one \(\alpha_i\), and four scheduling signals for the scheduling matrix \(\mbs{\Phi}_i\). Using three subcontrollers, one such set of scheduling matrices are 
\vspace{-4pt}
\begin{subequations} \label{eqn: matrix_scheduling_signals}
    \begin{align}
        \mbs{\Phi}_1(t) &= \begin{bmatrix} 
                            \mu_{1}s_1(t) + \nu_{1}s_2(t) & 0 \\ 
                            0                             & s_1(t) 
                        \end{bmatrix}, &\alpha_1 = 2,\\%
        \mbs{\Phi}_2(t) &= \begin{bmatrix} 
                            s_2(t) & 0 \\
                            s_2(t) & s_2(t) 
                        \end{bmatrix}, &\alpha_2 = 1,\\%
        \mbs{\Phi}_3(t) &= \begin{bmatrix} 
                            \mu_{2}s_3(t) + \nu_{2}s_2(t)   & 0 \\ 
                            0                               & s_3(t) 
                        \end{bmatrix}, &\alpha_3 = 2,
    \end{align} 
\end{subequations}
with \(\mu_{1} = 2\), \(\nu_{1} = 4\), \(\mu_{2} = 1\), and \(\nu_{2} = 2\), where \(s_1(t)\), \(s_2(t)\), and \(s_3(t)\) are defined in \cref{eq:scalar_scheduling_signals}. The scheduling matrices in \cref{eqn: matrix_scheduling_signals} are deliberately chosen to be diagonal or lower triangular, so that they are full rank, provided nonzero diagonal elements. It can easily be verified that \(\forall t \in \mathbb{R}_{\geq 0}\), \(\exists i \in \mathcal{N}\), such that \(\rank \mleft(\mbs{\Phi}_i(t) \mright) = n = 2\), provided \(\mu_1, \mu_2, \nu_1, \nu_2 > 0\). Furthermore, to highlight the impact of the off-diagonal element in \(\mbs{\Phi}_2(t)\) on \(\tau_2\), the second diagonal entry of each scheduling matrix is kept as its scalar counterpart, \(s_i(t)\).
Finally, output scheduling matrices are scaled by \(\alpha_i\), to demonstrate the effect of the scaling factor on the performance of the gain-scheduled controllers in \Cref{fig:GS}.
\begin{figure}[t]
    \vspace{10pt}
    \centering
    \includegraphics{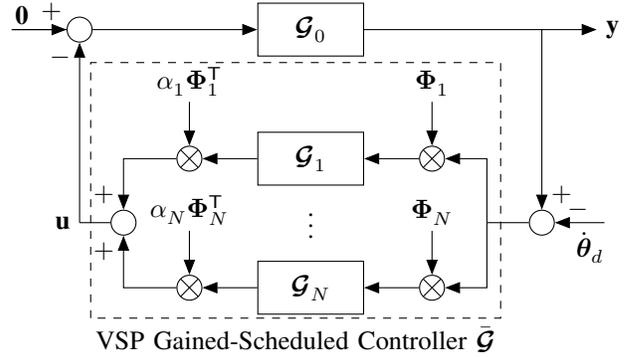}
    \vspace{-10pt}
    \caption{Gain-scheduled feedback control of the plant to be controlled \(\bm{\mathcal{G}}_0\), prewrapped with proportional control, and the gain-scheduled controller $\bar{\bm{\mathcal{G}}}$.}
    \label{fig:GS feedback control system}
    \vspace{-10pt}
\end{figure}
\subsection{Comparison}
Consider the rigid two-link planner robotic manipulator in \Cref{fig:robot} and its equation of motion \cref{eq:robot_dynamics}. Three different control approaches are compared with the objective of following the trajectory given by \cref{eqn:trajectory_interpolation}. As a baseline, a single VSP controller is designed about the linearization of the robot at the end of its trajectory. This corresponds to using the third linearization point \(\bar{\mbf{q}}_3\) in the linearized model \cref{eq:linearized_model}, and will be referred to as the unscheduled controller henceforth. The second approach, referred to as the scalar gain-scheduled (GS) controller, is presented in \cite{Forbes_Damaren}. In particular, \cite{Forbes_Damaren} gain-schedules three VSP subcontrollers, \(\bm{\mathcal{G}}_1\), \(\bm{\mathcal{G}}_2\), and \(\bm{\mathcal{G}}_3\), designed about the linearization points \(\bar{\mbf{q}}_1\), \(\bar{\mbf{q}}_2\), and \(\bar{\mbf{q}}_3\), using the scalar scheduling signals in \cref{eq:scalar_scheduling_signals}. These subcontrollers are then gain-scheduled as per \Cref{fig:GS}, in the parallel interconnection shown in \Cref{fig:GS feedback control system}, just with \(\mbs{\Phi}_i(t) = s_i(t) \eye\) and \(\alpha_i = 1\) for \(i \in \mathcal{N}\). The third approach, referred to as the matrix GS controller, only differs from the scalar gain-scheduled controller in that the scheduling matrices in \cref{eqn: matrix_scheduling_signals} are used instead. Note, across all three control approaches, the exact same \(\mbf{K}_{\text{p}}\), \(\mbf{Q}_{\text{LQR}}\), \(\mbf{R}_{\text{LQR}}\), and \(\delta\) are used for the synthesis of the VSP subcontrollers.
\afterpage{
    \begin{figure}[t!]
        \vspace{6pt}
        \centering
        \includegraphics[width=0.485\textwidth]{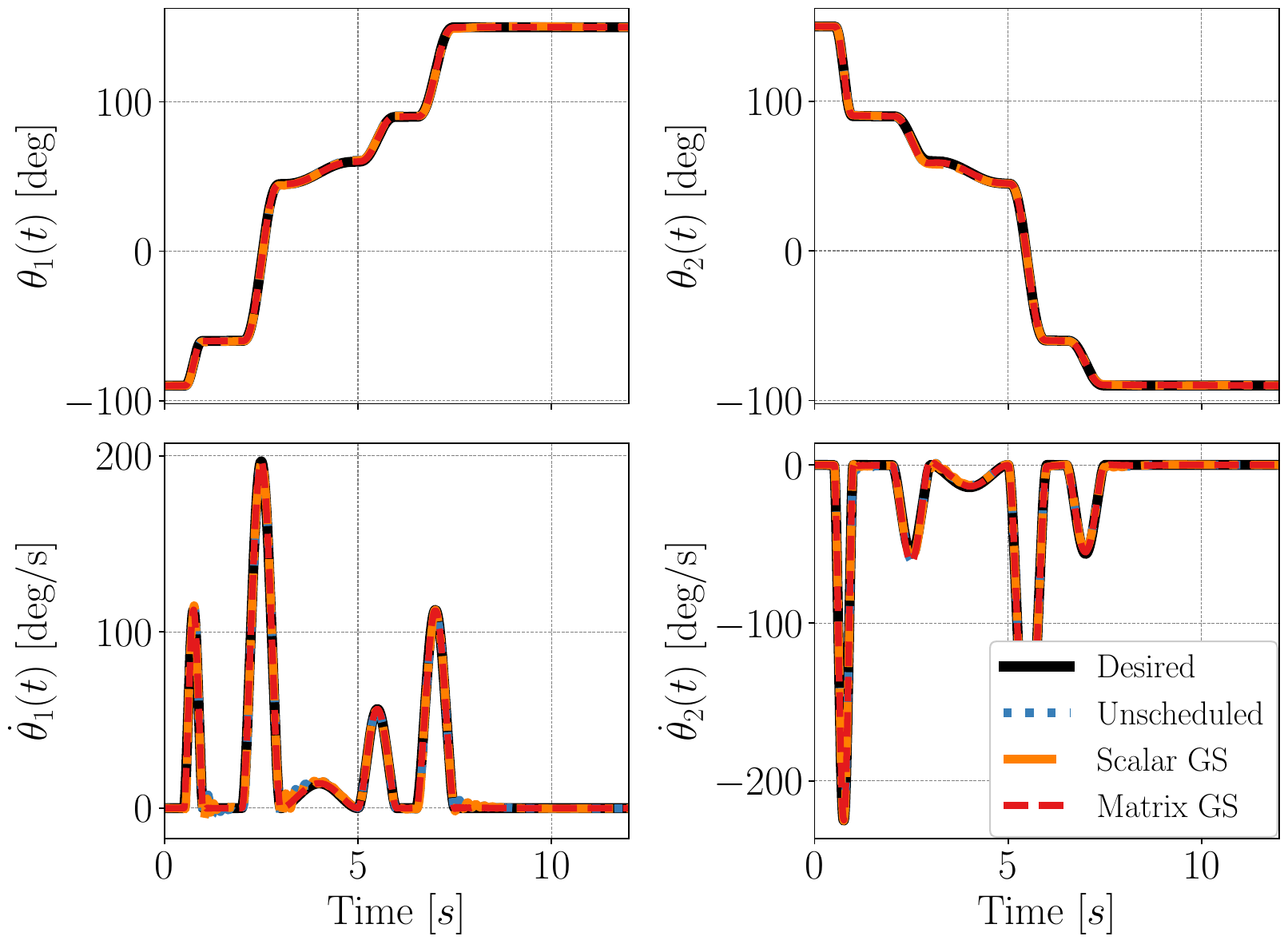}
        \vspace{-15pt}
        \caption{Comparison of joint angles and joint rates.}
        \vspace{20pt}
        \label{fig:trajectory comparison}
        \includegraphics[width=0.485\textwidth]{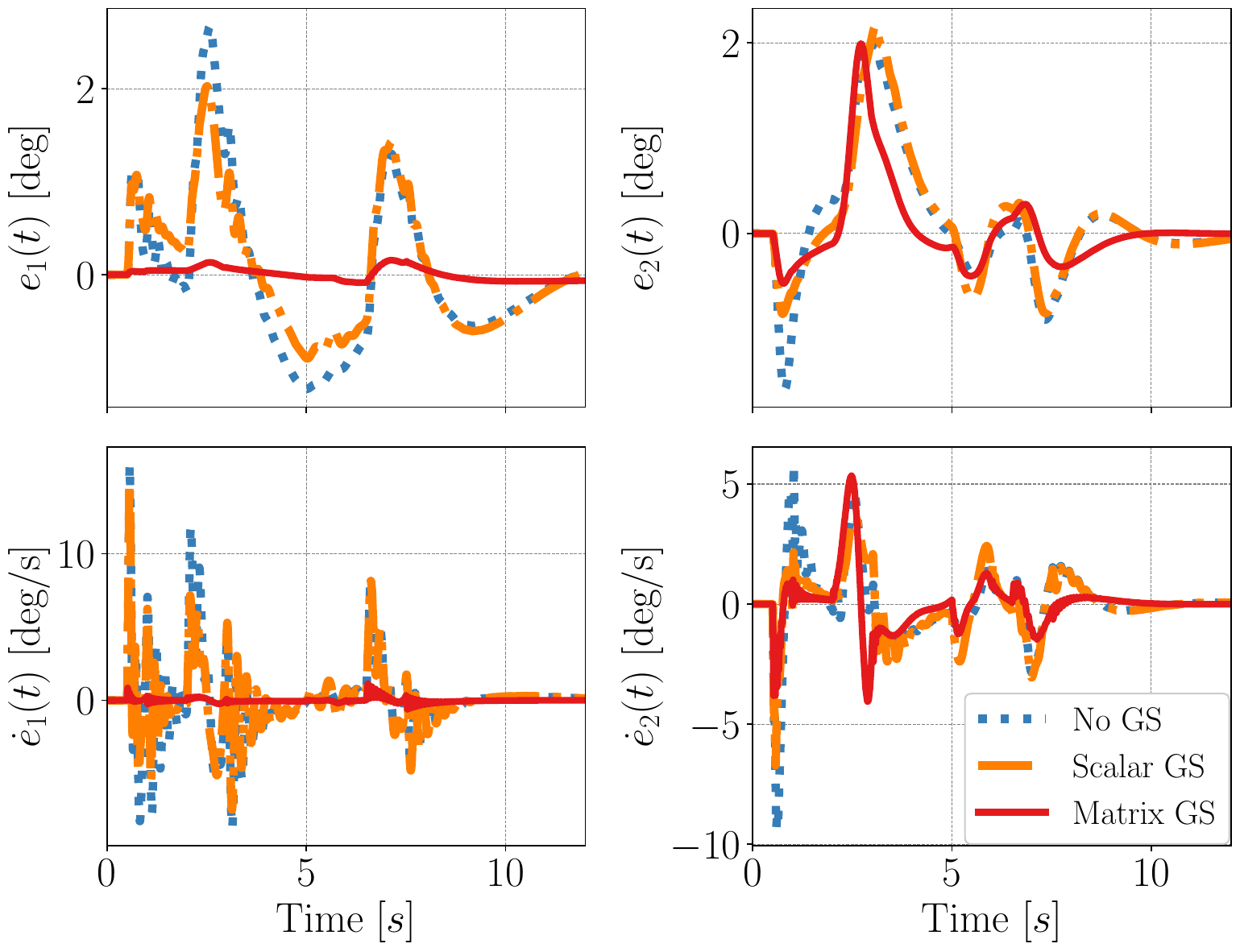}
        \vspace{-15pt}
        \caption{Comparison of joint angles errors and error rates.}
        \vspace{20pt}
        \label{fig:error comparison}
        \includegraphics[width=0.485\textwidth]{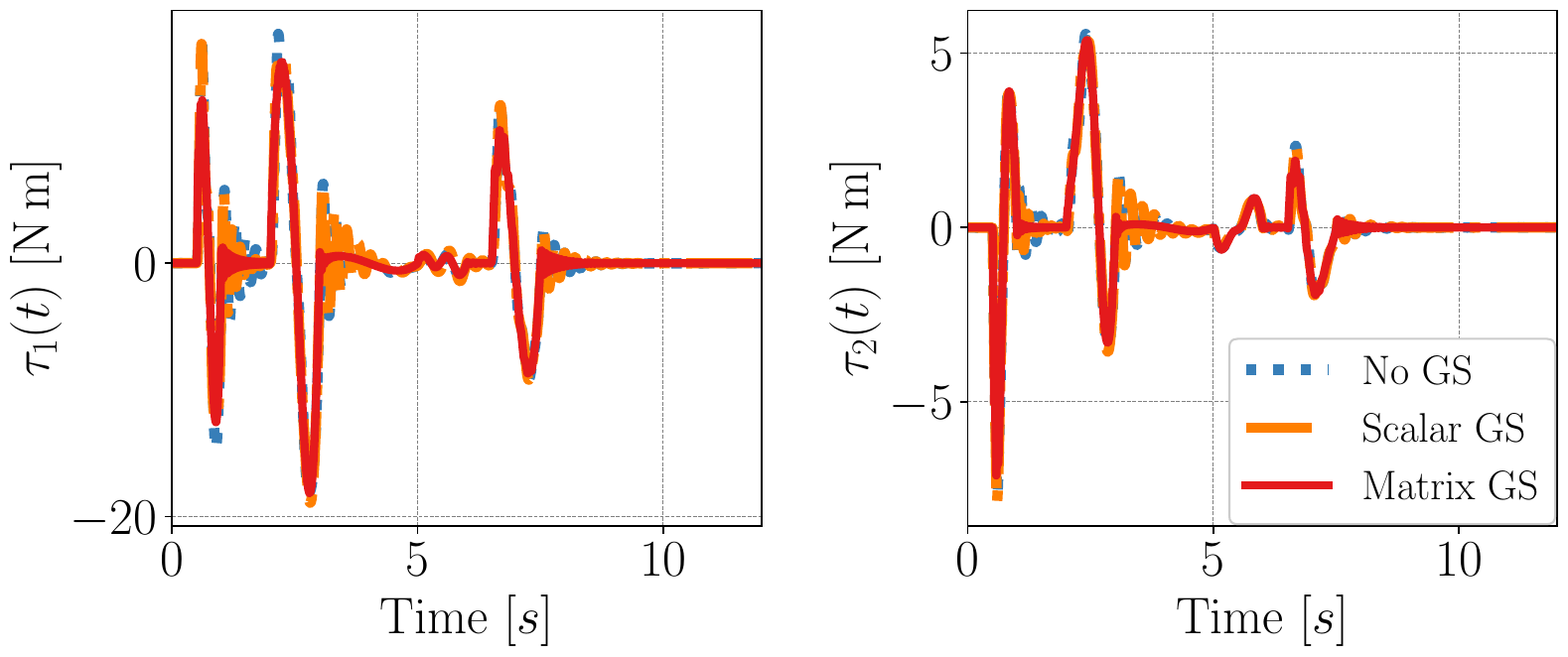}
        \vspace{-15pt}
        \caption{Comparison of joint torques.}
        \label{fig:torques comparison}
        \vspace{-20pt}
    \end{figure}
}
    \begin{table}[b]
        \vspace{-10pt}
        \caption{RMS Error of Joint Angle and Joint Angle Rate}
        \label{tab:RMS error}
        \vspace{-7pt}
        \small
        \begin{tabularx}{\linewidth}{l *{4}{S[table-format=1.4]}}
            \hline
            \hline
            & 
            \multicolumn{2}{c}{RMS angle error}
            &
            \multicolumn{2}{c}{RMS angle rate error} 
            \\
            \multirow[b]{2}{*}{Control method} 
            & 
            \multicolumn{2}{c}{$\si{[deg]}$}
            &
            \multicolumn{2}{c}{$\si{[deg/s]}$} 
            \\
            \cmidrule(l){2-3}
            \cmidrule(l){4-5}
            & {$e_1$} & {$e_2$} & {$\dot{e}_1$} & {$\dot{e}_2$} \\
            \hline
            Unscheduled        & 0.8328          & 0.6688          & 2.5933          & 1.5587 \\
            Scalar scheduling  & 0.6839          & 0.6464          & 2.1307          & 1.2702 \\
            Matrix scheduling  & \textbf{0.0668} & \textbf{0.4515} & \textbf{0.1480} & \textbf{1.1352} \\
            \hline
            \hline
        \end{tabularx}
    \end{table}
The desired trajectory \cref{eqn:trajectory_interpolation} along with the close tracking performance of the three subcontrollers are shown in \Cref{fig:trajectory comparison}. The joint angle error, \(\mbf{e}(t) = \begin{bmatrix} e_1(t) & e _2(t) \end{bmatrix}^{\trans} = \bm{\theta}_\text{d}(t) - \mbf{q}(t)\), is shown in \Cref{fig:error comparison}, where the matrix GS controller has noticeably less error, magnitude wise, than the scalar GS controller. The joint torques, \(\mbf{u}(t) = \begin{bmatrix} \tau_1(t) & \tau_2(t)\end{bmatrix}^{\trans}\), are shown in \Cref{fig:torques comparison}, where all three approaches have similar profiles. The root-mean-square (RMS) joint angle error, and joint angle error rates are tabulated in \Cref{tab:RMS error}. Again, the matrix GS controller realizes much lower RMS angle error and  RMS angle rate error. The \texttt{python} code used to generate the figures presented in this section can be found in the GitHub repository at \url{https://github.com/decargroup/matrix_scheduling_vsp_controllers}.
    \section{Closing Remarks}\label{sec: closing remarks}
Gain-scheduled control of VSP subcontrollers using scheduling matrices is considered in this paper. The proposed gain-scheduling architecture is shown to preserve the VSP properties of the subcontrollers, provided the scheduling matrices are bounded and strongly active, as defined in \Cref{sub: scheduling matrix properties}. The conditions on the scheduling matrices reduce to the same conditions on the scheduling signals reported in \cite{Damaren_passive_map,Forbes_Damaren} when the scheduling matrices are deliberately chosen to be scalar's times the identity matrix. The proposed gain-scheduling architecture is used to control a rigid two-link robot in simulation subject to model uncertainty. Numerical results highlight the added benefit of using scheduling matrices relative to scheduling signals.
    \printbibliography
\end{document}